\documentclass[a4paper,12pt]{article}

\usepackage{amsmath,amsthm,amssymb,mathtools}
\usepackage{mathrsfs}
\usepackage[top=2.2cm,bottom=2.3cm,left=2.7cm,right=2.7cm]{geometry}
\usepackage{setspace}
\usepackage[round,authoryear]{natbib}
\usepackage{booktabs}
\usepackage{array}
\usepackage{enumitem}
\usepackage{hyperref}

\hypersetup{
  colorlinks=true,
  linkcolor=blue,
  citecolor=blue,
  urlcolor=blue
}

\newtheorem{theorem}{Theorem}

\newtheorem{lemma}{Lemma}

\newtheorem{definition}{Definition}
\newtheorem{remark}{Remark}
\newtheorem{example}{Example}

\newcommand{\Stable}{\mathcal S}
\newcommand{\Reach}{\mathcal R}
\newcommand{\Hist}{\mathscr H}
\newcommand{\emptysetmatch}{\varnothing}

\title{Stable Matchings and Their Histories:\\
Historical Reachability under Population Shocks}
\author{Yi-You Yang\thanks{Department of Applied Mathematics, Chung Yuan Christian University, Taoyuan City, Taiwan. E-mail address: yyyang@cycu.edu.tw}}
\date{}

\begin{document}

\maketitle

\begin{abstract}
Sequential-entry procedures need not reach every stable matching.  We ask
whether full reachability is restored when agents on both sides may enter and
exit repeatedly, with proposal-chain restabilization after each population
change.  It is not: a three-by-three marriage market has a stable matching
that is unreachable from the empty active market along any such history
restricted to terminal agents.

We then allow temporary historical agents who are absent at the terminal
date.  For every target stable matching, we construct a target-dependent,
preference-preserving augmentation and an admissible history that reaches it,
with all temporary agents exiting before termination.  The construction
extends to many-to-one markets with responsive hospital preferences.

Thus reachability depends on the admissible history class.  Terminal-agent
histories may select a proper subset of the stable set, whereas
preference-preserving augmented histories recover the full stable set.  When
such augmented histories are admissible but past participants are unobserved,
the terminal primitives alone do not exclude any stable matching.
\end{abstract}


\section{Introduction}
\label{sec:introduction}

A stable matching is defined from the agents and preferences present in a
market.  Actual matching markets, however, are formed through histories.
Workers retire, graduates enter, positions open and close, and inherited
relationships are repeatedly adjusted after population changes.  This raises
a question that cannot be answered from the terminal market alone: which of
its stable matchings can be supported by a preceding population history?

The question is related to two classical strands of matching theory.
\citet{RothVandeVate1990} show that successive blocking-pair adjustments lead
from any initial matching to some stable matching.  Their result concerns
convergence to the stable set rather than reachability of a prescribed element
of that set.  The random order mechanism studied by \citet{Ma1996} introduces
agents sequentially and restabilizes the market after each arrival.
\citet{Ma1996} shows that not every stable matching need be generated.
\citet{KlausKlijn2007} correct part of the original argument while preserving
this conclusion, and \citet{Cheng2016} shows that deciding whether a specified
stable matching is reachable under the random order mechanism is NP-complete.

A second strand studies re-equilibration after population shocks.
\citet{BlumRothRothblum1997} analyze vacancy chains following retirements and
the opening of new positions in senior-level labor markets.  Subsequent work
examines timing, entry, and restabilization in marriage and many-to-one
markets \citep{BlumRothblum2002,Cantala2004,BoyleEchenique2009}.  These papers
study a forward problem: given an inherited stable matching and a population
shock, which stable outcome follows?  We study the associated inverse
problem: given a stable matching of the terminal market, does there exist a
population history leading to it?

We distinguish histories according to the agents they may contain.  A
\emph{terminal-agent history} uses only agents who are present in the terminal
market.  Agents on either side may enter and exit repeatedly, and every
population change is followed by the proposal-chain continuation specified
below.  Allowing both entry and exit permits repeated reversals in scarcity
orientation.  Nevertheless, our first result shows that such reversals do not
restore full reachability.  We construct a three-by-three marriage market
with a stable matching that cannot be reached from the empty active market by
any terminal-agent history.  The obstruction appears in the final transition:
the target matching has no stable five-agent predecessor whose missing agent,
upon entry, generates it through the prescribed proposal chain.

An \emph{augmented history} may additionally contain temporary agents who are
absent from the terminal market.  Their introduction must preserve each
original agent's preferences over the original opposite-side agents and the
unmatched option, and all temporary agents must leave before the terminal
date.  Our second result establishes
historical completeness under this broader class.  For every marriage market
and every target stable matching, there exists a target-dependent,
preference-preserving augmentation and an admissible history that terminates
at that matching.

The construction uses temporary agents as a historical scaffold.  Original
agents are first matched to private temporary partners.  These partners are
then removed sequentially, and stability of the target matching ensures that
each released agent is rejected by all preferred but incompatible partners
and accepted by the prescribed target partner.  The construction preserves
stability after every population event and removes all temporary agents before
the terminal date.

Our third result extends historical completeness to many-to-one markets with
responsive hospital preferences.  Private temporary hospitals initially hold
the original doctors, while temporary doctors occupy the seats of the
original hospitals.  Responsiveness and stability provide a ranking interval
in which these temporary seat holders can be placed: they lie below the
doctors assigned to the hospital in the target matching and above every doctor
who would prefer to deviate there.  The temporary doctors can consequently be
replaced one at a time by the target doctors without disturbing previously
installed assignments.

Formally, there exists a marriage market \(M\) such that
\[
  \Reach_{\Hist^{\mathrm{term}}(M)}(M)
  \subsetneq
  \Stable(M).
\]
By contrast, every responsive many-to-one market \(M\) satisfies
\[
  \Reach_{\Hist^{\mathrm{aug}}(M)}(M)
  =
  \Stable(M).
\]  Reachability is therefore a refinement of stability conditional on
a specified class of histories.  When the relevant population history is
known, it may eliminate stable outcomes.  When preference-preserving augmented histories are admissible but past
participants and their preferences are unobserved, the terminal primitives
alone do not justify deleting any element of the stable set.

The paper is related to, but distinct from, reachability through blocking-pair
dynamics.  \citet{AbeledoRothblum1995} and \citet{Rudov2024} allow unstable
intermediate matchings in a fixed population, whereas we require stability
after every population event.  \citet{Mao2026} studies whether existing
matched pairs can be preserved when a market expands or contracts.  We
instead use temporary expansions to construct a history supporting a
prescribed terminal stable matching.  The positive results here are
existential and target dependent; they do not provide one common augmentation
generating the entire stable set, nor do they cover general substitutable
hospital choice.

The remainder of the paper is organized as follows.
Section~\ref{sec:model} defines proposal-chain histories and historical
completeness.  Section~\ref{sec:negative} establishes nonreachability under
terminal-agent histories.  Section~\ref{sec:marriage} proves historical
completeness for marriage markets, and Section~\ref{sec:manytoone} extends the
construction to responsive many-to-one markets.
Section~\ref{sec:discussion} discusses the interpretation and limitations of
the results.

\section{Markets and Historical Reachability}
\label{sec:model}

\subsection{Marriage markets}

A marriage market is a tuple
\[
  M=(A,B,\succ),
\]
where \(A\) and \(B\) are finite disjoint sets and every agent has a strict
preference ordering over the agents on the opposite side and the unmatched
option \(\emptyset\).  A partner is acceptable if it is preferred to
\(\emptyset\).  A matching \(\mu\) assigns each agent either an acceptable
partner on the opposite side or \(\emptyset\), with
\(\mu(a)=b\) if and only if \(\mu(b)=a\).

A matching is \emph{stable} if it is individually rational and there is no
pair \((a,b)\in A\times B\) such that
\[
  b\succ_a\mu(a)
  \qquad\text{and}\qquad
  a\succ_b\mu(b).
\]
The stable set is denoted by \(\Stable(M)\).  Stable matchings exist by
\citet{GaleShapley1962}.  Their lattice structure is reviewed in
\citet{RothSotomayor1990}.

\subsection{Proposal-chain continuation}

We use the sequential restabilization process underlying the random-order and
vacancy-chain literatures.  Suppose the active market is stable and one agent
enters.  The entrant becomes the active proposer and applies to acceptable
partners in descending order, excluding partners who have already rejected
that proposer during the current chain.  A recipient accepts the proposal if
and only if she prefers the proposer to her current partner.  Otherwise she
rejects it and the proposer continues to the next acceptable partner.  An
acceptance displaces the recipient's previous partner, if any; the displaced
agent becomes the next proposer on the same side.  The process stops when a
proposal is accepted without displacing another agent or when the active
proposer has exhausted all acceptable partners.  The symmetric process applies
when the entrant belongs to the other side.

If a matched agent exits, the abandoned partner initiates the corresponding
proposal chain.  If the inherited matching restricted to the new population is
already stable, the continuation is allowed to have length zero.  These are
the single-agent entry and exit versions of the decentralized continuation
processes studied by \citet{Ma1996} and \citet{BlumRothRothblum1997}.  We only
use histories for which the specified continuation is finite and reaches a
stable matching before the next population event.

\begin{definition}[Stable proposal-chain history]
A stable proposal-chain history is a finite sequence
\[
  h=\bigl(M_t,\mu_t\bigr)_{t=0}^{T}
\]
such that:
\begin{enumerate}[label=(\roman*)]
  \item \(M_0\) is the empty active market and
        \(\mu_0=\emptysetmatch\);
  \item \(\mu_t\in\Stable(M_t)\) for every \(t\);
  \item \(M_{t+1}\) is obtained from \(M_t\) by the entry or exit of one
        agent;
  \item \(\mu_{t+1}\) is the stable outcome of the proposal-chain
        continuation from the inherited matching after that population
        event.
\end{enumerate}
\end{definition}

Let \(M\) be a fixed terminal market.  A \emph{terminal-agent history} of
\(M\) is a stable proposal-chain history in which every \(M_t\) is an induced
submarket of \(M\) and \(M_T=M\).  An \emph{augmented history} of \(M\) is
supported by a finite augmented market \(\widetilde M\): every \(M_t\) is an
induced submarket of \(\widetilde M\), \(M_T=M\), and every agent in
\(\widetilde M\setminus M\) is absent at the terminal date.  Thus an agent's
preferences are fixed throughout the history, including across repeated exits
and re-entries.  We write \(\Hist^{\mathrm{term}}(M)\) and
\(\Hist^{\mathrm{aug}}(M)\) for the corresponding classes of histories.

\begin{definition}[Scarcity orientation]
For a market with sides \(A\) and \(B\), a population event is
\emph{\(A\)-scarcity oriented} if an \(A\)-agent exits or a \(B\)-agent enters.
It is \emph{\(B\)-scarcity oriented} if a \(B\)-agent exits or an \(A\)-agent
enters.  A scarcity-orientation reversal occurs when two successive population
events have opposite orientations.
\end{definition}

\begin{definition}[Preference-preserving augmentation]
An augmented market \(\widetilde M\) of \(M\) is preference preserving if,
for every original agent, the restriction of the extended preference order to
the original opposite-side agents and \(\emptyset\) coincides with that
agent's preference order in \(M\).  New agents may be inserted anywhere in
the resulting preference lists, and their own preferences may be chosen
freely.  In a many-to-one market, the capacities of all original hospitals
are unchanged.  After all temporary agents exit, the terminal market is
exactly \(M\).
\end{definition}

For a class \(\Hist\) of admissible histories ending at \(M\), define
\[
  \Reach_{\Hist}(M)
  :=
  \left\{
    \mu\in\Stable(M):
    \text{some }h\in\Hist\text{ terminates at }(M,\mu)
  \right\}.
\]

\begin{definition}[Historical completeness]
A class of matching markets is historically complete under a class of
histories if
\[
  \Reach_{\Hist}(M)=\Stable(M)
\]
for every market \(M\) in the class.
\end{definition}

Historical completeness is an existential property.  The supporting
augmentation and the order of population events may depend on the target
stable matching.

\section{Repeated Entry and Exit Need Not Restore Full Reachability}
\label{sec:negative}

This section strengthens the nonreachability phenomenon of \citet{Ma1996}.
The agents observed in the terminal market may enter and exit an arbitrarily
large finite number of times, and the scarcity orientation may reverse
repeatedly.  Even so, some
stable matchings remain unreachable.

\begin{example}[A cyclic three-by-three market]
\label{ex:cycle}
Let \(A=\{m_1,m_2,m_3\}\) and \(B=\{w_1,w_2,w_3\}\), with preferences
\[
\begin{array}{lll}
 m_1: w_1\succ w_2\succ w_3,
& m_2: w_2\succ w_3\succ w_1,
& m_3: w_3\succ w_1\succ w_2,\\[1mm]
 w_1: m_2\succ m_3\succ m_1,
& w_2: m_3\succ m_1\succ m_2,
& w_3: m_1\succ m_2\succ m_3.
\end{array}
\]
All partners are acceptable.  The full market has exactly three stable
matchings:
\[
\begin{aligned}
  \mu^A
  &=\{(m_1,w_1),(m_2,w_2),(m_3,w_3)\},\\
  \mu^0
  &=\{(m_1,w_2),(m_2,w_3),(m_3,w_1)\},\\
  \mu^B
  &=\{(m_1,w_3),(m_2,w_1),(m_3,w_2)\}.
\end{aligned}
\]
The matching \(\mu^A\) is optimal for the men, \(\mu^B\) is optimal for the
women, and every agent receives the second-ranked partner at \(\mu^0\).
\end{example}

The following table records the unique stable matching of every five-agent
submarket and the result when the missing agent re-enters and initiates the
proposal chain.

\begin{center}
\renewcommand{\arraystretch}{1.18}
\begin{tabular}{>{\centering\arraybackslash}p{2.2cm}
                p{7.2cm}
                >{\centering\arraybackslash}p{2.7cm}}
\toprule
Missing agent & Unique stable matching before re-entry & Continuation after re-entry\\
\midrule
\(m_1\) & \(\{(m_2,w_2),(m_3,w_3)\}\) & \(\mu^A\)\\
\(m_2\) & \(\{(m_1,w_1),(m_3,w_3)\}\) & \(\mu^A\)\\
\(m_3\) & \(\{(m_1,w_1),(m_2,w_2)\}\) & \(\mu^A\)\\
\(w_1\) & \(\{(m_1,w_3),(m_3,w_2)\}\) & \(\mu^B\)\\
\(w_2\) & \(\{(m_1,w_3),(m_2,w_1)\}\) & \(\mu^B\)\\
\(w_3\) & \(\{(m_2,w_1),(m_3,w_2)\}\) & \(\mu^B\)\\
\bottomrule
\end{tabular}
\end{center}

The assertions in the example can be checked directly.  Since all partners
are acceptable and the two sides have equal size, every stable matching of the
full market is perfect.  The three perfect matchings not displayed in
Example~\ref{ex:cycle} are unstable:
\[
\begin{array}{c|c}
\text{matching} & \text{blocking pair}\\
\hline
\{(m_1,w_1),(m_2,w_3),(m_3,w_2)\} & (m_3,w_1)\\
\{(m_1,w_2),(m_2,w_1),(m_3,w_3)\} & (m_2,w_3)\\
\{(m_1,w_3),(m_2,w_2),(m_3,w_1)\} & (m_1,w_2).
\end{array}
\]
Hence the full market has exactly the three stated stable matchings.  For each
five-agent submarket, deferred acceptance with either side proposing yields
the same matching displayed in the table, so that matching is unique.  By
cyclic symmetry it is enough to check one missing man and one missing woman.
If \(m_1\) is absent, both versions of deferred acceptance yield
\(\{(m_2,w_2),(m_3,w_3)\}\); when \(m_1\) enters, he is accepted by his
first-choice unmatched woman \(w_1\), producing \(\mu^A\).  If \(w_1\) is
absent, both versions yield \(\{(m_1,w_3),(m_3,w_2)\}\); when \(w_1\) enters,
she is accepted by her first-choice unmatched man \(m_2\), producing
\(\mu^B\).  The remaining four rows follow by the same cyclic relabeling.

\begin{theorem}[Failure of terminal-agent historical completeness]
\label{thm:negative}
There exists a marriage market \(M\) such that
\[
  \Reach_{\Hist^{\mathrm{term}}}(M)
  \subsetneq
  \Stable(M),
\]
even when every terminal agent may enter and exit an arbitrarily large finite
number of times and every population change is followed by the specified
proposal-chain continuation.
\end{theorem}

\begin{proof}
Use the market in Example~\ref{ex:cycle}.  We show that \(\mu^0\) is not
reachable.

Suppose, to the contrary, that a terminal-agent history reaches \(\mu^0\), and
consider the first date at which the active market is the full six-agent
market and the stable matching is \(\mu^0\).  Since the history starts from the
empty active market and uses no agents outside the terminal population, the
immediately preceding population event that creates the full market must be
the entry of one of the six agents.  The preceding active market is therefore
the corresponding five-agent submarket.

Each five-agent submarket has the unique stable matching displayed in the
table.  If a man is the final entrant, his proposal chain produces \(\mu^A\).
If a woman is the final entrant, her proposal chain produces \(\mu^B\).
Neither case produces \(\mu^0\), a contradiction.

The argument is unaffected by earlier exits and re-entries.  Every visit to
the full market must be preceded by the re-entry of a missing terminal agent,
and the same six incoming transitions apply.
\end{proof}

\begin{remark}
Allowing arbitrarily many finite orientation reversals does not remove the
obstruction.  What fails is the existence of a stable codimension-one
predecessor whose entrant-initiated continuation leads to the target
matching.
\end{remark}

\section{Historical Completeness in Marriage Markets}
\label{sec:marriage}

We now allow agents who participated in the past but are absent from the
terminal market.  Under augmented histories, marriage markets are
historically complete.

\begin{theorem}[Historical completeness of marriage markets]
\label{thm:marriage-positive}
Let \(M=(A,B,\succ)\) be a marriage market and let
\(\mu\in\Stable(M)\).  There exists a finite preference-preserving augmentation
\(\widetilde M\) and an augmented stable proposal-chain history from the empty
active market that terminates at \((M,\mu)\).  All temporary agents have exited
at the terminal date.
\end{theorem}

\begin{proof}
For each man \(m\in A\), introduce a private temporary woman \(a_m\).  Place
\(a_m\) first in \(m\)'s extended preference order, above \(\emptyset\) and
every original woman.  The woman \(a_m\) finds only \(m\) acceptable, and
every temporary woman other than \(a_m\) is unacceptable to \(m\).

For each woman \(w\in B\), introduce a private temporary man \(b_w\), who
finds only \(w\) acceptable; every temporary man other than \(b_w\) is
unacceptable to \(w\).  If \(\mu(w)=m_w\in A\), insert \(b_w\) immediately
below \(m_w\) in \(w\)'s extended preference order.  Since \(\mu\) is
individually rational, this places \(b_w\) above \(\emptyset\).  If
\(\mu(w)=\emptyset\), insert \(b_w\) immediately above \(\emptyset\).
The relative rankings of all original agents are unchanged.

We first construct the scaffold matching
\[
  \nu_{\emptyset}
  =
  \{(m,a_m):m\in A\}
  \cup
  \{(b_w,w):w\in B\}.
\]
Starting from the empty active market, let the original women enter first,
then the temporary men \(b_w\), then the temporary women \(a_m\), and finally
the original men.  Each \(b_w\) is accepted by his unique acceptable woman.
Each \(a_m\) remains unmatched until \(m\) enters, at which point \(m\)
proposes first to \(a_m\) and is accepted.  After every entry the current
matching is stable: every entered original man has his first choice, each
matched temporary agent has its unique acceptable partner, and an unmatched
\(a_m\) accepts only a man who has not yet entered.  The final matching of this
phase is \(\nu_{\emptyset}\).

For \(S\subseteq A\), let \(\nu_S\) denote the matching in which
\[
  \nu_S(m)=
  \begin{cases}
    \mu(m),&m\in S,\\
    a_m,&m\notin S,
  \end{cases}
\]
where \(\nu_S(m)=\emptyset\) when \(m\in S\) and
\(\mu(m)=\emptyset\).  An original woman \(w\) is matched to \(\mu(w)\) if
\(\mu(w)\in S\), and otherwise to \(b_w\).  A temporary man displaced from
\(w\) is unmatched.

We show inductively that, after precisely the private women indexed by
\(S\) have exited, the current stable matching is \(\nu_S\).  The claim holds
for \(S=\emptyset\).  Take \(m\notin S\) and remove \(a_m\).  Suppose first
that \(\mu(m)=w^*\in B\).  For every woman \(w\succ_m w^*\), stability of
\(\mu\) implies that either \(m\) is unacceptable to \(w\) or
\(\mu(w)\succ_w m\).  If \(\mu(w)\in S\), then \(w\) is currently matched to
\(\mu(w)\) and rejects \(m\).  If \(\mu(w)\notin S\), then \(w\) is matched
to \(b_w\).  When \(\mu(w)\) is an original man, the insertion of \(b_w\)
immediately below \(\mu(w)\) gives \(b_w\succ_w m\).  When
\(\mu(w)=\emptyset\), stability implies \(\emptyset\succ_w m\), and hence
\(b_w\succ_w\emptyset\succ_w m\).  Thus every woman preferred by \(m\) to
\(w^*\) rejects him.

The woman \(w^*=\mu(m)\) is currently matched to \(b_{w^*}\), because her
target man \(m\) has not yet been released.  By construction,
\(m\succ_{w^*}b_{w^*}\), so she accepts \(m\) and displaces \(b_{w^*}\).
The displaced temporary man has just been rejected by his only acceptable
woman, so the chain ends.  The resulting matching is
\(\nu_{S\cup\{m\}}\).

If \(\mu(m)=\emptyset\), consider any original woman \(w\) whom \(m\) finds
acceptable.  Stability of \(\mu\) implies that \(w\) either finds \(m\)
unacceptable or has a target partner she prefers to \(m\).  In the latter case her current partner is either that target man or
her temporary holder \(b_w\), both of whom she prefers to \(m\).  Hence every
acceptable proposal is rejected and \(m\) remains unmatched.  Again the
resulting matching is \(\nu_{S\cup\{m\}}\).

It remains to verify stability of the induction invariant.  An unreleased man
has his first choice \(a_m\).  If a released man prefers an original woman to
his assignment under \(\mu\), the preceding argument shows that her current
partner is preferred to him.  A private temporary woman accepts only her
assigned man.  A temporary man \(b_w\) is either matched to \(w\), or has been
displaced by \(\mu(w)\), whom \(w\) prefers to \(b_w\).  Thus \(\nu_S\) is
stable for every \(S\).

After all private women have exited, the restriction to the original agents
is exactly \(\mu\).  Remove first every unmatched temporary man; this changes
neither the matching nor stability.  Any temporary man still matched is
paired with a woman \(w\) satisfying \(\mu(w)=\emptyset\).  Remove these men
one at a time.  After such an exit, \(w\) becomes unmatched.  No original man
forms a blocking pair with \(w\), because \(\mu\) is stable, and no other
temporary man finds \(w\) acceptable.  Hence every such exit has a zero-length
stable continuation.  All temporary agents have now left, and the terminal
market and matching are \((M,\mu)\).
\end{proof}

\begin{remark}[Target dependence]
The augmentation in Theorem~\ref{thm:marriage-positive} depends on the target
matching \(\mu\), because each temporary holder is inserted relative to the
target partner.  The theorem establishes
\[
  \forall \mu\in\Stable(M)\;\exists \widetilde M_\mu\;\exists h_\mu,
\]
not the stronger existence of one common augmentation that reaches every
stable matching through different histories.
\end{remark}

\section{Responsive Many-to-One Markets}
\label{sec:manytoone}

We now consider a college-admissions or doctor--hospital market.  Let \(D\) be
a finite set of doctors and \(H\) a finite set of hospitals.  Hospital \(h\)
has capacity \(q_h\geq 1\).  Each doctor has a strict preference over
\(H\cup\{\emptyset\}\).  Each hospital has a strict ranking over
\(D\cup\{\emptyset\}\), inducing a responsive preference over feasible sets:
it chooses up to \(q_h\) highest-ranked acceptable doctors.  See
\citet{RothSotomayor1990} and \citet{ChambersYenmez2018}.

A matching \(\mu\) assigns each doctor to at most one hospital and each
hospital a set \(\mu(h)\) of at most \(q_h\) doctors.  Stability is individual
rationality together with the absence of a blocking pair \((d,h)\), where
\[
  h\succ_d\mu(d)
  \quad\text{and}\quad
  d\in C_h\bigl(\mu(h)\cup\{d\}\bigr).
\]

The histories constructed in this section require only doctor-initiated
proposal chains.  Such a chain begins when a doctor enters or when her private
capacity-one temporary hospital exits.  The doctor applies to acceptable
hospitals in descending order.  Upon receiving a proposal, hospital \(h\)
retains its \(q_h\) highest-ranked acceptable doctors among its current
assignees and the proposer.  A displaced doctor becomes the next active
proposer.  The chain stops when no doctor is displaced or the active doctor
has exhausted all acceptable hospitals.  Every other population event used
below leaves the inherited matching stable and therefore induces a zero-length
continuation.

For a fixed stable matching \(\mu\), define the potential deviators to
hospital \(h\) by
\[
  B_h(\mu)
  :=
  \left\{
    d\in D\setminus\mu(h):
    h\succ_d\mu(d)
    \text{ and }d\succ_h\emptyset
  \right\}.
\]

\begin{lemma}[Separation at a stable matching]
\label{lem:separation}
Let \(\mu\) be stable.  For every hospital \(h\):
\begin{enumerate}[label=(\roman*)]
  \item if \(B_h(\mu)\neq\emptyset\), then \(|\mu(h)|=q_h\);
  \item every doctor in \(\mu(h)\) is ranked above every doctor in
        \(B_h(\mu)\).
\end{enumerate}
\end{lemma}

\begin{proof}
If \(d\in B_h(\mu)\) and \(|\mu(h)|<q_h\), responsiveness implies that
\(h\) would accept \(d\), so \((d,h)\) would block \(\mu\).  Hence
\(|\mu(h)|=q_h\).  If \(h\) ranked \(d\in B_h(\mu)\) above some
\(d'\in\mu(h)\), it would replace \(d'\) by \(d\), and \((d,h)\) would again
block \(\mu\).
\end{proof}

The lemma provides an interval in each hospital's ranking in which temporary
seat holders can be inserted.

\begin{theorem}[Historical completeness under responsiveness]
\label{thm:responsive}
Let \(M=(D,H,q,\succ)\) be a many-to-one matching market with responsive
hospital preferences, and let \(\mu\in\Stable(M)\).  There exists a finite
preference-preserving augmentation and an augmented stable proposal-chain
history from the empty active market that terminates at \((M,\mu)\).  The
construction uses \(|D|\) temporary capacity-one hospitals and
\(\sum_{h\in H}q_h\) temporary doctors.
\end{theorem}

\begin{proof}
\textit{Step 1: Construct the temporary agents.}
For every original doctor \(d\in D\), introduce a private temporary hospital
\(a_d\) with capacity one.  Place \(a_d\) first in \(d\)'s extended
preference order, above \(\emptyset\) and every original hospital.  Every
other temporary hospital is unacceptable to \(d\), and \(a_d\) finds only
\(d\) acceptable.

For every original hospital \(h\), introduce \(q_h\) temporary doctors
\[
  b_{h,1},\ldots,b_{h,q_h}.
\]
Each \(b_{h,j}\) finds only \(h\) acceptable, and temporary doctors assigned
to other original hospitals are unacceptable to \(h\).  Insert the doctors
\(b_{h,1},\ldots,b_{h,q_h}\) into \(h\)'s ranking so that
\[
  d'\succ_h b_{h,1}\succ_h\cdots\succ_h b_{h,q_h}\succ_h d
\]
for every \(d'\in\mu(h)\) and every \(d\in B_h(\mu)\), and rank every
\(b_{h,j}\) above \(\emptyset\).  Lemma~\ref{lem:separation} guarantees that
such an insertion is possible without changing the relative ranking of any
two original doctors.  When \(B_h(\mu)=\emptyset\), place the temporary
doctors below all doctors in \(\mu(h)\) and above \(\emptyset\).

\textit{Step 2: Build the scaffold.}
Starting from the empty active market, let the original hospitals enter one at
a time.  Next let all temporary doctors enter.  Each \(b_{h,j}\) is accepted
by its unique acceptable hospital, so every original hospital becomes filled
to capacity.  Then let the private temporary hospitals enter, followed by the
original doctors.  Each original doctor proposes first to her private hospital
and is accepted.  The resulting matching is
\[
  \nu_{\emptyset}
  =
  \{(d,a_d):d\in D\}
  \cup
  \bigcup_{h\in H}\{(b_{h,j},h):j=1,\ldots,q_h\}.
\]
It is stable: every original doctor has her first choice, every matched
temporary agent has its unique acceptable partner, and an unmatched private
hospital accepts only a doctor who has not yet entered.  The same reasoning
shows stability after every entry in this phase.

\textit{Step 3: Release the original doctors.}
For \(S\subseteq D\), define \(\nu_S\) as follows.  Every doctor in \(S\) has
her assignment under \(\mu\), every doctor outside \(S\) is matched to her
private hospital, and each original hospital \(h\) holds
\[
  \mu(h)\cap S
\]
together with its highest
\[
  q_h-|\mu(h)\cap S|
\]
remaining temporary doctors.  Temporary doctors displaced from \(h\) are
unmatched.

We prove inductively that \(\nu_S\) is the current stable matching after
precisely the private hospitals indexed by \(S\) have exited.  The assertion
holds for \(S=\emptyset\).  Take \(d\notin S\) and remove \(a_d\).  Doctor
\(d\) becomes unmatched and initiates the doctor-proposing continuation.

Suppose first that \(\mu(d)=h^*\).  Consider a hospital \(h\succ_d h^*\).
If \(d\) is unacceptable to \(h\), it rejects her.  Otherwise
\(d\in B_h(\mu)\).  The hospital is currently filled by doctors in
\(\mu(h)\cap S\) and by its highest remaining temporary doctors.  Every such
doctor is ranked above \(d\), so \(h\) rejects her.  When \(d\) proposes to
\(h^*\), at least one temporary doctor remains there because
\(d\in\mu(h^*)\setminus S\).  The hospital ranks every doctor in
\(\mu(h^*)\), including \(d\), above every temporary doctor.  It therefore
accepts \(d\) and displaces its lowest remaining temporary doctor.  The
displaced doctor has just been rejected by her only acceptable hospital, so
the chain ends.  The resulting matching is \(\nu_{S\cup\{d\}}\).

If \(\mu(d)=\emptyset\), every acceptable hospital preferred by \(d\) to
being unmatched satisfies \(d\in B_h(\mu)\).  Each is full and all of its
current occupants are ranked above \(d\), so every proposal is rejected and
\(d\) remains unmatched.  The resulting matching is again
\(\nu_{S\cup\{d\}}\).

To verify the induction invariant, note first that an unreleased doctor has
her first choice.  If a released doctor prefers another hospital to her target
assignment, then either she is unacceptable there or she belongs to
\(B_h(\mu)\), in which case every current occupant is ranked above her.  A
displaced temporary doctor finds only her original hospital acceptable and,
because the hospital always displaces its lowest remaining temporary doctor,
is ranked below every doctor currently held there.  Private temporary
hospitals accept only their assigned doctors.  Hence \(\nu_S\) is stable for
every \(S\).

After all private hospitals have exited, the original doctors are matched
according to \(\mu\), while each hospital \(h\) holds exactly
\(q_h-|\mu(h)|\) temporary doctors when it is underfilled at \(\mu\).

\textit{Step 4: Remove the temporary doctors.}
First remove every temporary doctor who is currently unmatched.  This changes
neither the matching nor stability.  Then remove the remaining matched
temporary doctors one at a time.  Such doctors are present only at hospitals
underfilled in \(\mu\).  By stability of \(\mu\), no original doctor both
prefers such a hospital to her assignment and is acceptable to it; otherwise
that doctor and the hospital would block \(\mu\).  Temporary doctors assigned
to other hospitals do not find this hospital acceptable, and all displaced
temporary doctors have already exited.  Hence the inherited matching after
each removal is already stable, so the continuation has length zero.

All temporary agents have now exited, the active market is the original market
\(M\), and the terminal matching is \(\mu\).
\end{proof}

\section{Interpretation and Further Questions}
\label{sec:discussion}

\subsection{The stable set and historically reachable sets}

Three levels of information should be distinguished.  The stable set
\(\Stable(M)\) contains all stable outcomes supported by the terminal agents
and their preferences.  A class \(\Hist\) of admissible histories selects the
historically reachable set
\[
  \Reach_{\Hist}(M)
  =
  \left\{
    \mu\in\Stable(M):
    \text{some }h\in\Hist\text{ terminates at }(M,\mu)
  \right\}.
\]
A fully specified history, by contrast, terminates at a single outcome rather
than selecting a set.

Theorem~\ref{thm:negative} shows that a rich class of histories need not be
complete when the historical population is restricted to the agents visible
at the terminal date.  Theorems~\ref{thm:marriage-positive} and
\ref{thm:responsive} show that the full stable set is recovered when
preference-preserving, target-dependent augmentations by temporary historical
agents are admitted.  Thus, for some terminal markets,
\[
  \Reach_{\Hist^{\mathrm{term}}(M)}(M)
  \subsetneq
  \Stable(M)
  =
  \Reach_{\Hist^{\mathrm{aug}}(M)}(M).
\]

This supports an informational interpretation of the classical stable set.
It is the appropriate possibility set when the terminal market is observed
but the relevant population history is not.  A smaller reachable set is justified
only after the analyst imposes restrictions on the admissible histories.
Reachability is therefore a history-conditioned refinement of stability, not
a replacement for the stable set.

\subsection{Relation to scarcity orientation}

The definition of scarcity orientation separates the direction of a local
population event from the net change between two distant dates.  Because
terminal agents may enter and exit repeatedly, terminal-agent histories may
contain arbitrarily many orientation reversals.  Theorem~\ref{thm:negative}
shows that such reversals alone do not restore historical completeness: the
target matching in the example has no stable codimension-one predecessor whose
entrant-initiated continuation reaches it.

The positive results enlarge history in a different way.  Temporary agents
who are absent from the terminal market create rejection thresholds that guide
proposal chains toward a prescribed stable matching.  They form a historical
scaffold: they alter the route by which the market reaches the terminal
primitives without altering the relative rankings among the original agents
or remaining in the terminal market.  Hence a stable matching may be
unreachable from every history composed only of terminal agents and still be
consistent with a richer, unobserved population history.

\subsection{Limits of the present result}

The positive theorems deliberately grant substantial freedom.  The temporary
population and its preferences depend on the target matching.  The results do
not minimize the number of historical agents, and they do not construct one
common augmentation capable of generating the whole stable set through
different entry and exit orders.  Nor do they show that every admissible
restabilization path reaches the target; they construct one supporting
proposal-chain history.

Responsiveness also matters.  The proof uses a common ranking interval between
target doctors and all potential blockers.  Temporary doctors can be inserted
in that interval and replaced one at a time without displacing previously
installed target doctors.  Under a general substitutable choice function, the common ranking interval
used in the responsive proof is unavailable, and the present argument does
not establish a sequential placeholder representation.  A natural next step
is therefore to characterize choice functions satisfying a \emph{sequential
placeholder replacement property} and to determine whether that property is
necessary and sufficient for historical completeness.

Other extensions include many-to-many matching, matching with contracts,
capacities that change over time, and histories in which temporary agents'
preferences are restricted rather than freely designed.  The construction
also raises a quantitative question.  For a market \(M\) and target
\(\mu\in\Stable(M)\), define
\[
  \kappa_M(\mu)
  :=
  \min\{\text{number of distinct temporary agents in a supporting history}\}.
\]
Theorem~\ref{thm:responsive} gives the upper bound
\[
  \kappa_M(\mu)
  \leq
  |D|+\sum_{h\in H}q_h.
\]
Determining sharper bounds and their relation to the position of \(\mu\) in
the stable lattice is left for future work.

\section{Conclusion}

Reachability depends on what is allowed to count as market history.  Repeated
entry and exit of the agents observed at the terminal date need not make every
stable matching reachable under the specified proposal-chain continuation.
By contrast, when preference-preserving, target-dependent augmentations by
temporary historical agents are admitted, every stable matching of a marriage
market, and more generally every stable matching of a responsive many-to-one
market, has a supporting history.

The contrast clarifies the relation between static stability and historical
selection.  A known class of admissible histories may eliminate stable
outcomes that are incompatible with it.  When the relevant history is unknown
and preference-preserving augmented histories are admitted, however,
nonreachability under a more restrictive process cannot justify discarding
elements of the stable set.  The full stable set therefore remains the appropriate possibility set when
the terminal market is observed but the admissible population history is
unknown.

\bibliographystyle{plainnat}
\bibliography{historical_reachability_stable_matchings_final}

\end{document}